\documentclass[11pt]{article}
\usepackage{graphicx}
\usepackage{caption}
\usepackage{amsfonts}
\usepackage{amsmath}
\usepackage{amssymb}
\usepackage{appendix}
\usepackage{url}
\usepackage{fancyhdr}
\usepackage{indentfirst}
\usepackage{enumerate}
\usepackage{amsthm}
\usepackage{color}
\usepackage{natbib}
\usepackage[colorlinks=true,citecolor=blue]{hyperref}

 \usepackage{amsthm,natbib,comment}
\usepackage{verbatim}
\usepackage{float}
\usepackage{diagbox}

\def\linfcv{\buildrel \mathcal L^\infty \over \rightarrow}
\def\d{\mathrm{d}}

\DeclareMathOperator*{\argmin}{argmin}

\newcommand{\E}{\mathbb{E}}
\newcommand{\R}{\mathbb{R}}

\newcommand{\N}{\mathbb{N}}
\newcommand{\p}{\mathbb{P}}

\newcommand{\one}{\mathbf{1}}

\newcommand{\esssup}{\mathrm{ess\mbox{-}sup}}
\newcommand{\essinf}{\mathrm{ess\mbox{-}inf}}

\renewcommand{\leq}{\leqslant}
\renewcommand{\epsilon}{\varepsilon}
\theoremstyle{plain}
\newtheorem{theorem}{Theorem}

\newtheorem{proposition}[theorem]{Proposition}
\theoremstyle{definition}
\newtheorem{definition}{Definition}[section]
\newtheorem{example}{Example}[section]

\theoremstyle{remark}
\newtheorem{remark}{Remark}[section]
\theoremstyle{definition}

\numberwithin{equation}{section} \numberwithin{theorem}{section}
\renewcommand{\cite}{\citet}
\renewcommand{\cdots}{\dots}

\usepackage{setspace}

\usepackage{multirow}
\usepackage{setspace} 
\usepackage{footmisc}
\renewcommand{\cite}{\citet}

\usepackage[compact]{titlesec}
\titlespacing{\section}{0pt}{8pt}{4pt}
\titlespacing{\subsection}{0pt}{6pt}{4pt}
\titlespacing{\subsubsection}{0pt}{5pt}{3pt}

\renewcommand{\topmargin}{-1.6cm}
\begin{document}

\captionsetup[figure]{labelformat=simple, labelsep=period}
\captionsetup[table]{labelformat=simple, labelsep=period}

\title{Conditional generalized quantiles based on expected utility model and equivalent characterization of properties}

\author{ Qinyu Wu\thanks{wu05551@mail.ustc.edu.cn; Department of Statistics and Finance, School of Management,
University of Science and Technology of China,
Hefei, Anhui, China}
\quad
Fan Yang\thanks{fan.yang@uwaterloo.ca; Department of Statistics and Actuarial Science, University of Waterloo, Canada.}
\quad
Ping Zhang\thanks{zp012@mail.ustc.edu.cn; Department of Statistics and Finance, School of Management,
University of Science and Technology of China,
Hefei, Anhui, China}
}
\date{}
\maketitle{}

\abstract{
As a counterpart to the (static) risk measures of generalized quantiles and motivated by \cite{BB18}, we propose a new kind of conditional risk measure called conditional generalized quantiles.
We first  show  their well-definedness  and they can
be equivalently characterised by a conditional first order condition. We also  discuss their  main properties, and, especially, We give  the characterization of coherency/convexity.
For potential
applications  as a dynamic risk measure, we study their time consistency properties, and establish their equivalent characterizations among conditional generalized quantiles.
}

\maketitle

\section{Introduction}

Risk measures serve as quantitative tools to determine minimum capital reserves and have attracted growing interest since the seminal work of \cite{ADEH99}. 
Various (static) risk measures have been proposed, such as value-at-risk (VaR, also known as quantiles), which is the regulatory risk measure used in the Basel requirements. 
Generalized quantiles, as a nonlinear generalization of quantiles, were introduced by \cite{BB14} as a well-defined class of risk measures that include VaR and expectiles (\citealp{NP87}) as special cases. It is also well known that such measures are closely related to the Haezendonck-Goovaerts risk measures (\citealp{BR12}), $\psi$-mean certainty equivalent (\citealp{BT07}) and shortfall risk measures (\citealp{FS16}). 
Recall that the generalized quantile of a risk $X$ is defined as the minimizer of the minimization problem $\min \limits_{x\in\R}\pi _\alpha (X,x)$ (\citealp{BB14}), where
\begin{equation}\label{eq-pi}
\pi _\alpha (X,x) = \alpha \E[ u_1((X-x)^+)] +(1-\alpha)\E [u_2((X-x)^-)].
\end{equation}
Here, $a^+=\max\{a,0\}$, $a^-=\max\{-a,0\}$, $\alpha\in (0,1)$ is the confidence level to balance the shortfall risk $(X-x)^+$ and the over-required capital risk $(X-x)^-$, and $u_1$ and $u_2$ are two loss (disutility) functions in the expected utility model that are assumed to be strictly increasing convex functions with $u_1(0)=u_2(0)=0$ and $u_1(1)=u_2(1)=1$, and are used to
transform the two risks $(X-x)^+$ and  $(X-x)^-$ respectively. If $u_1(x)=u_2(x)=x$, then the generalized quantile reduces to the classic quantile, and if $u_1(x)=u_2(x)=x^2$, the generalized quantile reduces to the expectile.

Generalized quantiles have  wide applications in quantitative risk management and statistics, see \cite{G11}, \cite{EGJK16}, \cite{CS17} and references therein. It has been argued that the generalized quantile has specific merits in contexts where a decision maker (regulator) is  more concerned about upper-tail realizations of a loss random variable, and views the two risks $(X-x)^+$ and  $(X-x)^-$ asymmetrically with
underestimating losses being more detrimental than overestimating (\citealp{RRM14} and \citealp{MC18}).

We illustrate generalized quantiles in solvency risk management. For  a random loss $X$ faced by a financial institution, let $x$ be the capital reservation which will be determined by a regulator. If $X>x$, then $(X-x)^+=\max\{X-x,0\}$ represents
the ``shortfall risk", while if $x > X$, then $(X-x)^-=\max\{x-X, 0\}$ describes ``the over-required capital
risk'' under the required solvency capital decision $x$. Obviously, there is a trade-off between the shortfall
risk and the over-required capital risk.  It is generally believed that the shortfall risk $(X-x)^+$ will be transferred to society if defaulted, while the over-required capital risk $(X-x)^-$ represents a waste of capital/resource. The loss transferred to society may induce severe consequences (e.g., financial crisis) if the loss is large. From the perspective of the regulation,
  it is reasonable to assume that  the regulator uses different loss functions $u_1$ and $u_2$ to transform or quantitate $(X-x)^+$ and $(X-x)^-$ respectively.
In the meantime,  as a compromise or to take the overall risk of the decision
on a required solvency capital into consideration, the regulator may want to consider the weighted risk and would like to minimize the expectation of the weighted risk  \eqref{eq-pi}, where $\alpha$ and $1-\alpha$  can be viewed as the sensitive measures of the regulator to
the expectations of the shortfall risk and the over-required capital risk, respectively. 

In reality, most decision makers are making  decisions dynamically over time, usually at discrete times (e.g., once a day or once a week). Dynamic risk assessments are essential in the framework of multi-period problem. Hence, a natural extension of a static risk measure (which can be viewed as a risk assessment in a one-period problem) is given by a conditional risk measure, which takes account of the information available at the time of the risk assessment and is essential in dynamic risk measures (\citealp{AP11}; \citealp{FG04}; \citealp{RSE05}). Specifically,  the dynamical setting is described by some
filtered probability space and the risk assessment is updated over the time in accordance with the new information. That is,  similar to the concept of conditional expectation, the risk measure takes value in a space of random variables which is measurable with respect to some $\sigma$-field.  Conditional risk measures for classic risk measures have been widely studied.  Among others,
\cite{RS07} investigated the conditional tail VaR, and \cite{BB18} studied the conditional expectile.
In this paper, we introduce the notion of \emph{conditional generalized quantile}
(given in Section \ref{Conditional generalized quantile}). The properties of conditional generalized quantile as a dynamic risk is the first aim of the paper.
Specifically, we are interested in the minimization problem
\begin{align}\label{main-opt}
	 \min \limits_{Z\in \mathcal{L} ^{\infty}(\Omega,\mathcal{G},\p)} \pi _\alpha (X,Z),
	\end{align}
	where $\pi _\alpha (X,x)$ is defined by \eqref{eq-pi}, $\mathcal{G}$ is a sub-$\sigma$-algebra of $\mathcal{F}$ on $(\Omega,\mathcal{F})$ and $\mathcal{L} ^{\infty}(\Omega,\mathcal{G},\p)$ is the set of all bounded and measurable random variables on $(\Omega,\mathcal{G},\p)$.
We define  the minimizer in $\mathcal{L} ^{\infty}(\Omega,\mathcal{G},\p)$ as the conditional generalized quantile and study its properties as a dynamic risk measure. The contributions of this paper are summarized below.
\begin{itemize}
\item[(i)] It is nontrivial whether the minimizer of \eqref{main-opt} exists. To guarantee that the definition of conditional generalized quantile is reasonable (existence and uniqueness), we first prove that the set of minimizers is not empty, and then take the essential infimum of the set which is again a minimizer as the conditional generalized quantile (Theorem \ref{lemma-def}).
\item[(ii)] We introduce the concept of conditional shortfall risk measure which is induced by the classical shortfall risk measure (\cite{FS16}). The equivalence between the conditional generalized quantiles and conditional shortfall risk measures is established (Theorem $\ref{th-eq}$). Based on this, we discuss desired properties such as monotonicity, conditional translation invariance, conditional convexity and conditional positive homogeneity (Proposition $\ref{pro-basic}$). Moreover, we describes the necessary and sufficient condition for conditional generalized quantiles being a conditional convex or coherent risk measure (Theorem $\ref{th-cxcoherency}$).
\item[(iii)] For possible applications as dynamic risk measures, we discuss some kinds of time consistency properties that are sequential consistency, dynamic consistency and supermartingale property. We conclude that dynamic generalized quantiles is sequentially consistent, and dynamic entropic risk measure, which is a special case of dynamic generalized quantiles, is the unique class that is dynamic consistent or has supermartingale property (Theorem \ref{th-timeconsistency}).
\end{itemize}

The rest of the paper is organized as follows. In Section \ref{Preliminary}, we recall the definition and some desired properties of conditional risk measures.
In Section \ref{Conditional generalized quantile}, the formal definition of the conditional generalized quantile
is presented.
We also show that conditional generalized quantile can
be equivalently characterized by a conditional first order condition (Proposition \ref{pro-minimizer}).
Section \ref{Conditional shortfall risk measures}
provides the definition of the conditional shortfall risk measure and establishes the equivalence between the conditional generalized quantiles and conditional shortfall risk measures. In Section \ref{sec:characterization}, we study the properties of conditional generalized quantiles and the characterization of conditional convexity or coherency.
Section \ref{Properties} considers  conditional generalized quantiles in a dynamic framework and some further properties are investigated.


\emph{Throughout the paper},
$\mathcal U_{icx}^{+}$ is the set of all strictly increasing and convex functions  $u$ on $\R_{+}$ with $u(0) = 0$ and $u(1)=1$. 
We confine on the set of bounded real valued random variables defined on a non-atomic probability space $(\Omega,\mathcal{F},\p)$, and we adopt the convention 
that a positive (negative) value of $X$ means a financial loss (profit). All equalities, inequalities or convergence concerning random variables are under the sense of $\p \mbox{-}$a.s.. For simplicity, we denote by $\mathcal{L} ^{\infty}(\mathcal{F})=\mathcal{L} ^{\infty}(\Omega,\mathcal{F},\p)$ as the set of all bounded random variables that are measurable with respect to $\mathcal{F}$.

\section{Preliminary}\label{Preliminary}

In this section, we introduce the concept of conditional risk measures, which is not new, and list some desired properties. The notion of a dynamic risk measure as a collection of conditional risk measures was introduced by \cite{W99} and further studied via different formal approaches by \cite{A07}, \cite{R04}, \cite{Weber06}, \cite{RSE05}, \cite{S03}, \cite{FG04} and \cite{CDK06}. Most of these papers considered dynamic risk measures for discrete-time processes, but we prefer the work in \cite{DS05} instead simply for its use of random variables, as we use here.

Conditional risk measures interpret the situation when additional information $\mathcal{G}$ is provided for the assessment of the risk of a payoff $X$. In this situation, the risk measurement of $X$ leads to a random variable $\rho^\mathcal{G}(X)$. It is natural to consider $\rho^\mathcal{G}(X)(\omega)$ as the risk of $X$ when the state $\omega\in\Omega$ is available. First, we employ the definition of a conditional risk measure as in \cite{DS05}.
A \emph{conditional risk measure} is a functional $\rho^\mathcal{G}$ from $\mathcal{L}^\infty(\mathcal{F})$ to $\mathcal{L}^\infty(\mathcal{G})$\footnote{The original definition of a conditional risk measure is $\rho^\mathcal{G}:\mathcal{L}^\infty(\mathcal{F})\to \mathcal{L}^0(\mathcal{G})$, while we write $\rho^\mathcal{G}:\mathcal{L}^\infty(\mathcal{F})\to \mathcal{L}^\infty(\mathcal{G})$. This is because
all considered conditional risk measures in this paper are monotonic, conditional translation-invariant and normalized, and these properties guarantee that $\rho^\mathcal{G}(X)$ is bounded for all $X\in\mathcal{L}^\infty(\mathcal{F})$.}.
As a principle for regulation, we list some desired properties of conditional risk measures as follows.
\begin{enumerate}
\item[(P1)] {Monotonicity}: $\rho^\mathcal{G}(X)\leq \rho^\mathcal{G}(Y) $ for all $X, \, Y \in \mathcal L^\infty(\mathcal{F})$ such that $X\le Y$.
\item[(P2)] {Conditional translation-invariance}: $\rho^\mathcal{G}(X+Z)= \rho^\mathcal{G}(X)+ Z.$ for all $X \in \mathcal L^\infty(\mathcal{F})$ and $Z\in \mathcal L^\infty(\mathcal{G})$.
\item [(P3)] {Conditional convexity}: $\rho^\mathcal{G}(\Lambda X+ (1-\Lambda)Y)\le \Lambda\rho^\mathcal{G}(X)+ (1-\Lambda)\rho^\mathcal{G}(X) $ for all $X, \, Y \in \mathcal L^\infty(\mathcal{F})$ and $\Lambda\in \mathcal L^\infty(\mathcal{G})$ with $0 \le \Lambda \le1$.
\item[(P4)] {Normalization}: $\rho^\mathcal{G}(0)= 0$.
\item[(P5)] {Conditional positive homogeneity}: $\rho^\mathcal{G}(\Lambda X)= \Lambda\rho^\mathcal{G}(X)$ for all $X \in  L^\infty(\mathcal{F}) $ and $\Lambda\in\mathcal L^\infty(\mathcal{G})$ with $\Lambda \ge 0$.
\end{enumerate}

\begin{definition}\label{def2-1}
A conditional risk measure $\rho^\mathcal{G}:\mathcal{L}^\infty(\mathcal{F})\to \mathcal{L}^\infty(\mathcal{G})$ is said to be a {\em conditional convex risk measure} if it satisfies (P1), (P2), (P3) and (P4).
\end{definition}

\begin{definition}\label{def2-2}
A conditional risk measure is said to be a \emph{conditional coherent risk measure} if it is a conditional convex risk measure and satisfies (P5).
\end{definition}

\section{Conditional generalized quantiles }\label{Conditional generalized quantile}

In this section, we define the \emph{conditional generalized quantile} as the minimizer of the optimization problem \eqref{main-opt}.
Before presenting the formal definition of the conditional generalized quantile, we first show below that the definition is well {posed}
, that is, for any $X\in\mathcal L^{\infty}(\mathcal F)$, the minimization problem \eqref{main-opt}
admits a minimizer in $\mathcal L^{\infty}(\mathcal G)$. 
Denote by \begin{align}\label{phi}\phi_\alpha(X,Z):=\alpha u_1((X-Z)^+)+(1-\alpha)u_2((X-Z)^-),\end{align} and recall that the minimization problem \eqref{main-opt} is
\begin{align*}
	\min \limits_{Z\in \mathcal{L} ^{\infty}(\mathcal{G})} \pi _\alpha (X,Z)=\min \limits_{Z\in \mathcal{L} ^{\infty}(\mathcal{G})}\E [\phi_\alpha(X,Z)].
 \end{align*}
\begin{theorem}\label{lemma-def}
Let $X \in \mathcal{L} ^{\infty}(\mathcal{F})$, $\mathcal G \subseteq  \mathcal{F}$ be a $\sigma$-algebra, $u_1, u_2 \in \mathcal U_{icx}^{+}$, and $\alpha \in (0,1)$. The following statements hold.
\begin{enumerate}
\item[(i)] The set of  minimizers of the problem  \eqref{main-opt} is not empty, that is, $\argmin \limits_{Z\in \mathcal{L} ^{\infty}(\mathcal{G})} \pi _\alpha (X,Z)\neq \emptyset$.

\item[(ii)] If $Z_1,Z_2\in\argmin \limits_{Z\in \mathcal{L} ^{\infty}(\mathcal{G})} \pi _\alpha (X,Z)$, then
$$
\E[\phi _\alpha (X,Z_1)|\mathcal G]=\E[\phi _\alpha (X,Z_2)|\mathcal G].
$$
\item[(iii)] $\essinf\argmin \limits_{Z\in \mathcal{L} ^{\infty}(\mathcal{G})} \pi _\alpha (X,Z)\in \argmin \limits_{Z\in \mathcal{L} ^{\infty}(\mathcal{G})} \pi _\alpha (X,Z)$.

\end{enumerate}
\end{theorem}

\begin{proof}
(i) We employ the following steps to show the result.
 \begin{enumerate} [1.]
 \item We first show that for any given $X$, there exists $m>0$ such that
\begin{align} \label{eq-39-1}
	\argmin \limits_{Z\in \mathcal{L}^{\infty}(\mathcal{G})} \pi _\alpha (X,Z) =
	\argmin \limits_{Z\in \mathcal{L}^{\infty}_m(\mathcal{G})} \pi _\alpha (X,Z),
 \end{align}
where $\mathcal{L}^{\infty}_m(\mathcal{G})=\{Z\in \mathcal{L}^{\infty}(\mathcal{G}): \|Z\|_\infty\le m\}$. To see \eqref{eq-39-1}, note that there exists $m\in (0,\infty)$ such that $\lVert X\rVert _\infty \le m$.
Let $ Z\in \mathcal{L}^\infty(\mathcal G)$ be such that $\lVert Z\rVert _\infty >m$, which implies that $\p(A)>0$ with $A=\{|Z|>m \}\in \mathcal{G}$. Define $Z^*$   as
$$Z^*= -m\one_{\{Z<-m\}} + Z\one_{A^c} + m\one_{\{Z>m\}},
$$
where $\one_A$ represents the indicator function of $A$.
One can verify that $Z^*\in \mathcal G$ and $\lVert Z^*\rVert _\infty \le m $. Note further that
\begin{align*}
	\pi _\alpha (X,Z^*)- \pi _\alpha (X,Z) & = \alpha \E u_1((X-Z^*)^+)+ (1-\alpha)\E u_2((X-Z^*)^-)\\
	 &
	~~~ -  \alpha \E u_1((X-Z)^+)- (1-\alpha)\E u_2((X-Z)^-) \\
	& = (1-\alpha)\E \{[u_2((X-m)^-)- u_2[((X-Z)^-)]\one_{\{Z>m \}}\}\\
	&~~~ + \alpha \E \{[u_1((X+m)^+)\one_{\{Z<-m \}}
	- u_1((X-Z)^+)]\one_{\{Z<-m \}}\}\\
	&< 0,
\end{align*}
where the last inequality holds as $u_1,u_2$ are strictly increasing and $\p(A)>0$. Hence, we have that \eqref{eq-39-1} holds.

\item  By \eqref{eq-39-1} and $X\in \mathcal L^\infty$,  we have that   $Z\mapsto\pi _\alpha (X,Z)$ satisfies the lower semicontinuity.
That is, for any bounded sequence $\{Z_n,n\in\N\}$ which converges to some $Z~\p \mbox{-}a.s.$ as $n\to\infty$,  we have $\pi _\alpha (X,Z)\le\liminf\limits_{n\to\infty}\pi _\alpha (X,Z_n)$.

\item We next show that the conditions in Theorem 7.3.1 of \cite{KZ05}
are all satisfied if  $\mathcal{L} ^{\infty}_m(\mathcal{G})$ is endowed with the weak* topology.
\begin{itemize}
\item [(a)]	Weak* compactness. By Banach-Alaoglu's Theorem, we have $\mathcal{L}^{\infty}_m(\mathcal{G})$ is weak* compact.

\item [(b)] Weak* lower semicontinuity. To prove the weak* lower semicontinuity of $\pi _\alpha (X,Z)$ with respect to $Z\in \mathcal L_m^\infty(\mathcal G)$, it suffices to show that $\mathcal{B}:=\{Z\in \mathcal{L}^{\infty}_m(\mathcal{G}):\pi _\alpha (X,Z)\le c\}$ is weak* closed for $c\in \mathbb{R}$. First noting that $\pi_\alpha(x,z):= \alpha u_1((x-z)^+)+(1-\alpha) u_2((x-z)^-)$ is convex in $z$ for each $x\in \R $, we have that $\pi _\alpha (X,Z)$ is convex in $Z$, and thus, $\mathcal B$ is convex.  Hence, by Lemma A.65 in \cite{FS16}, it suffices to show $\mathcal B_r$  is a closed set in $\mathcal L^1$, where
$$\mathcal B_r:=\mathcal B\cap \{Z\in \mathcal{L}^{\infty}(\mathcal{G}): \|Z\|_\infty\le r\},~~r>0.$$

Let $\{Z_n\}$ be a sequence in $\mathcal{B}_r$ converging to some random variable $Z$ in $\mathcal{L}^1$, then by Skorohod Theorem, there exists a subsequence $\{Z_{nk}\}$  converges to $Z$ $\p \mbox{-}a.s.$  By the lower semicontinuity of $Z\mapsto\pi _\alpha (X,Z)$  with respect to $\p \mbox{-}a.s.$ convergence, we have that $\pi _\alpha (X,Z)\le\liminf\limits_{n\to\infty}\pi _\alpha (X,Z_n)$. Obviously, $\|Z\|_\infty\le \min\{r,m\}$ as $Z_n\in \mathcal{L}^{\infty}_m(\mathcal{G})$ and $\|Z_n\|_\infty\le \min\{r,m\}$ for all $ n\in\N.$ Hence, we have $Z\in \mathcal B_r$, that is, $\mathcal B_r$ is a closed set in $\mathcal L^1$. Therefore, we  have  $\mathcal B$ is weak* closed, and thus, $Z\mapsto\pi _\alpha (X,Z)$ is weak* lower semicontinuous.
\end{itemize}
Combining (a) and (b), we have the conditions in Theorem 7.3.1 of \cite{KZ05}
are all satisfied, and thus, $\argmin \limits_{Z\in \mathcal{L} ^{\infty}(\mathcal{G})} \pi _\alpha (X,Z)$ is not empty. This completes the proof of (i).
\end{enumerate}

(ii)
We prove this statement by contradiction. Let $Z_1,Z_2\in\argmin \limits_{Z\in \mathcal{L} ^{\infty}(\mathcal{G})} \pi _\alpha (X,Z)$, and define
$$
Y_1=\E[\phi_\alpha(X,Z_1)|\mathcal G]~~{\rm and}~~ Y_2=\E[\phi_\alpha(X,Z_2)|\mathcal G].
$$
We assume now that $\p(Y_1> Y_2)>0$. Let $Z=Z_1\one_{\{Y_1\le Y_2\}}+Z_2\one_{\{Y_1> Y_2\}}$. Note that $Z_1,Z_2,Y_1,Y_2$ all belong to $\mathcal{L} ^{\infty}(\mathcal{G})$, and thus $Z\in \mathcal{L} ^{\infty}(\mathcal{G})$. On the other hand,
\begin{align*}
\pi_\alpha(X,Z)&=\E[\E[\phi_\alpha(X,Z)|\mathcal G]]\\
&=\E[\E[\phi_{\alpha}(X,Z_1)\one_{\{Y_1\le Y_2\}}|\mathcal G]+\E[\phi_{\alpha}(X,Z_2)\one_{\{Y_1> Y_2\}}|\mathcal G]]\\
&=\E[\one_{\{Y_1\le Y_2\}}\E[\phi_{\alpha}(X,Z_1)|\mathcal G]+\one_{\{Y_1> Y_2\}}\E[\phi_{\alpha}(X,Z_2)|\mathcal G]]\\
&=\E[Y_1\one_{\{Y_1\le Y_2\}}+Y_2\one_{\{Y_1> Y_2\}}]\\
&=\E[Y_1\land Y_2] <\E[Y_1]=\pi_\alpha(X,Z_1),
\end{align*}
yielding a contradiction. Similarly, if $\p(Y_1<Y_2)>0$, one can obtain $\pi_\alpha(X,Z)<\pi_\alpha(X,Z_1)$. Hence, we verify the desired result.

(iii) Finally, we show that the essential infimum of the set of minimizers belongs to the set.
Define $\mathcal{B}:= \argmin_{Z\in \mathcal{L} ^{\infty}(\mathcal{G})} \pi _\alpha (X,Z)$ and $Z^*:= \essinf ~\mathcal{B}$. Applying Theorem A.33 in \cite{FS16} and Lemma 3 in the Appendix of \cite{DS05}, we prove $Z^* \in \mathcal{B}$ by equivalently showing that $\mathcal{B}$ is directed downward (named here for convenience), which is the opposite of the definition of directed upward.

For any $Z_1, Z_2 \in \mathcal{B}$,  
let $A:=\left\{Z_1 \le Z_2\right\}$, which satisfies $A\in \mathcal{G}$, and define $Z^\prime=\min\{Z_1,Z_2\}=Z_1\one_{A}+Z_2\one_{A^c}$. We aim to show that $Z'\in \mathcal{B}$. To see it, on the one hand, one can check that
$Z^\prime \in\mathcal L^{\infty}(\mathcal G)$.
On the other hand,
\begin{align*}
	\pi _\alpha (X, Z^\prime) &= \alpha \E[ u_1((X- Z_1\one_{A}-Z_2\one_{A^c})^+) ] +(1-\alpha) \E[ u_2((X-Z_1\one_{A}-Z_2\one_{A^c})^-)]\\
& = \alpha\E [\one_{A} u_1((X- Z_1)^+)] +\alpha\E[\one_{A^c} u_1((X-Z_2)^+)]\\
	&+(1-\alpha)\E [\one_{A} u_2((X- Z_1)^-)] +(1-\alpha)\E[\one_{A^c} u_2((X-Z_2)^-)]\\
&=\alpha\E [\one_{A} \E[u_1((X- Z_1)^+)|\mathcal G]] +\alpha\E[\one_{A^c} \E[u_1((X-Z_2)^+)|\mathcal G]]\\
	&+(1-\alpha)\E [\one_{A} \E[u_2((X- Z_1)^-)|\mathcal G]] +(1-\alpha)\E[\one_{A^c} \E[u_2((X-Z_2)^-)|\mathcal G]]\\
&=\E[\one_A\E[\phi_\alpha(X,Z_1)|\mathcal G]]+\E[\one_{A^c}\E[\phi_\alpha(X,Z_2)|\mathcal G]]\\
& =\E[\E[\phi_\alpha(X,Z_1)|\mathcal G]]\\
&=\pi_\alpha(X,Z_1),
\end{align*}
where the penultimate equality follows from statement (ii). Hence, we have $Z'\in\mathcal B$, and $\mathcal B$ is directed downward. This completes the proof.


\end{proof}

We now present the formal definition of conditional generalized  quantiles.
\begin{definition}\label{def3-1}
Let $\mathcal G \subseteq  \mathcal{F}$ be a $\sigma$-algebra, $u_1, u_2 \in \mathcal U_{icx}^{+}$ and $\alpha \in (0,1)$. The \emph{conditional generalized quantile of $X \in \mathcal{L} ^{\infty}(\mathcal{F})$, denoted by $\rho_\alpha ^\mathcal G(X)$}, is defined as
\begin{align*}
	\rho_\alpha ^\mathcal G(X) = \essinf\argmin \limits_{Z\in \mathcal{L} ^{\infty}(\mathcal{G})} \pi _\alpha (X,Z).
 \end{align*}

\end{definition}
The operater ``$\essinf$" in the definition is used to give a unique random variable as the conditional generalized quantile.
\footnote{ Note that the risk measure $\rho_\alpha ^\mathcal G(X)$ is the left endpoint of the nonempty closed
interval $\argmin \limits_{Z\in \mathcal{L} ^{\infty}(\mathcal{G})} \pi _\alpha (X,Z)$ and we focus on the ``smallest'' version here.}
If $u_1(x)=u_2(x)=x^2$, $x\in\R$, we have that the conditional generalized quantile reduces to the conditional expectile proposed by \cite{BB18}, which is the unique minimizer of the following minimization problem
$$\min \limits_{Z\in\mathcal{L}^1(\mathcal{G})}\{\alpha \E[((X-Z)^{+})^2]+(1-\alpha) \E[((X-Z)^{-})^2]\}.$$

In the following proposition, we study the minimizer of the minimization problem \eqref{main-opt}.
\begin{proposition}\label{pro-minimizer}
Let $X \in \mathcal{L} ^{\infty}(\mathcal{F})$, $\mathcal G \subseteq  \mathcal{F}$, $u_1, u_2 \in \mathcal U_{icx}^{+}$, and $\alpha \in (0,1)$. For the minimization problem \eqref{main-opt},
\begin{enumerate}
\item[(i)] $Z^*\in \mathcal{L} ^{\infty}(\mathcal{G})$ is a minimizer if and only if it satisfies
\begin{align}
	\label{eq:0424-1}	\left\{
		\begin{array}{lcl}
		\alpha \E [(u_1)_-^\prime((X-Z^*)^+)|\mathcal{G}]\le(1-\alpha)\E[(u_2)_+^\prime((X-Z^*)^-)|\mathcal{G}],&\\
		\alpha \E[(u_1)_+^\prime((X-Z^*)^+)|\mathcal{G}]\ge(1-\alpha)\E[(u_2)_-^\prime((X-Z^*)^-)|\mathcal{G}],&\\
		\end{array}
		\right.
		\end{align}
where $(u_1)_-^\prime(x)=\frac{\partial^-}{\partial x}u_1(x)\one_{\{x>0\}}$, $(u_1)_+^\prime(x)=\frac{\partial^+}{\partial x}u_1(x)\one_{\{x\ge 0\}}$, $(u_2)_+^\prime(x)=\frac{\partial^+}{\partial x}u_2(x)\one_{\{x\ge0\}}$ and $(u_2)_-^\prime(x)=\frac{\partial^-}{\partial x}u_2(x)\one_{\{x>0\}}.$
\item[(ii)]  If additionally $u_1$, $u_2$ are differentiable with $u_1^\prime(0)= u_2^\prime(0)=0$, then $Z^*\in \mathcal{L} ^{\infty}(\mathcal{G})$ is a minimizer if and only if it satisfies
$$	\alpha \E [u_1^\prime(X-Z^*)^+|\mathcal{G}]=(1-\alpha)\E [u_2^\prime(X-Z^*)^-|\mathcal{G}].
$$
\end{enumerate}
\end{proposition}

\begin{proof}
It is clear that (ii) is the special case of (i), and we present the proof of (i) here.
For $A\in\mathcal{G}$, define $$f_A(t):= \alpha \E [u_1((X-(Z^*+t\one_A))^+)] +(1-\alpha)\E [u_2((X-(Z^*+t\one_A))^-)].$$
By the dominated convergence theorem, $f_A$ always has left and right derivatives. We have
\begin{align*}
	\begin{array}{lcl}
	(f_A)_+^\prime(t)= -\alpha \E [(u_1)_-^\prime((X-(Z^*+t\one_A))^+)\one_A] +(1-\alpha)\E [(u_2)_+^\prime((X-(Z^*+t\one_A))^-)\one_A],&\\
	(f_A)_-^\prime(t)= -\alpha \E [(u_1)_+^\prime((X-(Z^*+t\one_A))^+)\one_A] +(1-\alpha)\E [(u_2)_-^\prime((X-(Z^*+t\one_A))^-)\one_A].&\\
	\end{array}
	\end{align*}
Since $Z^*$ is a minimizer if and only if $(f_A)_+^\prime(0)\ge0 \ge (f_A)_-^\prime(0)$, we have that for each $A\in\mathcal{G}$,
\begin{align}
\label{eq:0414-2}	\left\{
	\begin{array}{lcl}
	\alpha \E [(u_1)_-^\prime((X-Z^*)^+)\one_A]\le (1-\alpha)\E [(u_2)_+^\prime((X-Z^*)^-)\one_A]&\\
	\alpha \E[(u_1)_+^\prime((X-Z^*)^+)\one_A]\ge (1-\alpha)\E[(u_2)_-^\prime((X-Z^*)^-)\one_A].&\\
	\end{array}
	\right.
	\end{align}
This is equivalent to \eqref{eq:0424-1}. %
We thus complete the proof of   (i).
\end{proof}
\begin{remark} Under the assumptions of Lemma \ref{lemma-def}, if $u_1$ and $u_2$ are strictly convex, then $\phi_\alpha(x,z)$ is strictly convex in $z$ for every $x$, and the minimization problem \eqref{main-opt} has a unique minimizer in the sense that $X^*=Y^*$ $\p \mbox{-}a.s.$ if $X^*\in \mathcal L^\infty(\mathcal G)$ and $Y^*\in \mathcal L^\infty(\mathcal G)$ are both minimizers. In this case, the operation $\essinf$ can be omitted.
To see it, let $Z_1,Z_2\in\argmin \limits_{Z\in \mathcal{L} ^{\infty}(\mathcal{G})} \pi _\alpha (X,Z)$. Suppose that $Z_1\neq Z_2$, a.s.  There exist $\delta>0$ and $A_0\in\mathcal{G}$ with $\p(A_0)>0$ such that $Z_1>Z_2+\delta$ on $A_0$. By Proposition \ref{pro-minimizer} (i), we have that \eqref{eq:0414-2} holds for $A=A_0$ and $Z^*$ replaced by $Z_i$, $i=1,2$.
 Note that $\phi _\alpha (X,Z)$ is strictly convex in $Z$, $(u_1)_-^\prime,(u_1)_+^\prime,(u_2)_-^\prime$ and $(u_2)_+^\prime$ are strictly increasing functions. We have that on $A_0$,
\begin{align*}
\alpha(u_1)_+^\prime((X-Z_1)^+)& < \alpha(u_1)_+^\prime((X-(Z_2+\delta))^+)
 <\alpha(u_1)_-^\prime((X-Z_2)^+)\\
& \leq(1-\alpha)(u_2)_+^\prime((X-Z_2)^-)
<(1-\alpha)(u_2)_-^\prime((X-(Z_2+\delta))^-)\\
& < (1-\alpha)(u_2)_-^\prime((X- Z_1)^-).
\end{align*}
This contradicts $\alpha \E[(u_1)_+^\prime((X-Z_1)^+)\one_{A_0}]\ge (1-\alpha)\E[(u_2)_-^\prime((X-Z_1)^-)\one_{A_0}]$. Thus we conclude that the minimizer is  unique.
\end{remark}

\section{Conditional shortfall risk measures}
 \label{Conditional shortfall risk measures}
We introduce the definition of another class of risk measures, called conditional shortfall risk measures, and then show that it is a one-to-one mapping from the class of conditional generalized quantiles to the class of conditional shortfall risk measures. Denote by $\mathcal V=\{v: \R\to\R| v~\text{is~increasing},~v(0-)\le 0~{\rm and}~v(0+)\ge 0 \}$.
\begin{definition}\label{def3-2}
Let $X \in \mathcal{L} ^{\infty}(\mathcal{F})$, $\mathcal G \subseteq  \mathcal{F}$ be a $\sigma$-algebra, and $v\in\mathcal V$. We define the \emph{conditional shortfall risk measure}\footnote{The conditional shortfall risk measures were first mentioned by \cite{Weber06}. However, the paper focused on the characterization of the shortfall risk measure and did not study the properties of the risk measure as a dynamic risk measure.} as
\begin{align*}\label{def3}
	\rho_{v} ^\mathcal G(X) &= \essinf\left\{Z\in \mathcal{L} ^{\infty}(\mathcal{G}):\E [v(X-Z)|\mathcal{G}]\le 0 \right\}.
	\end{align*}
\end{definition}
\begin{theorem}\label{th-eq}
\begin{enumerate}
\item[(i)]Let $X \in \mathcal{L} ^{\infty}(\mathcal{F})$, $\mathcal G \subseteq  \mathcal{F}$ be a $\sigma$-algebra, $u_1, u_2 \in \mathcal U_{icx}^{+}$ and $\alpha \in (0,1)$, $\rho_\alpha ^\mathcal G(X)$ be defined as in Definition $\ref{def3-1}$. We have
 $$\rho_\alpha ^\mathcal G(X) = \rho_{v} ^\mathcal G(X),$$ where $\rho_v ^\mathcal G(X)$ is a conditional shortfall defined as in Definition $\ref{def3-2}$ with
\begin{equation}\label{eq-v}
v(x) = \begin{cases}
		\alpha (u_1)_-^\prime(x), &x>0,\\
		-(1-\alpha) (u_2)_+^\prime(-x) ,&x\le 0.
		\end{cases}
\end{equation}
\item[(ii)] Let $X \in \mathcal{L} ^{\infty}(\mathcal{F})$, $\mathcal G \subseteq  \mathcal{F}$ be a $\sigma$-algebra and $v \in \mathcal V$, $\rho_v ^\mathcal G(X)$ be defined as in Definition $\ref{def3-2}$. We have $$\rho_{v} ^\mathcal G(X)= \rho_\alpha ^\mathcal G(X),$$ where $\rho_\alpha ^\mathcal G(X)$ is a conditional generalized quantile defined as in Definition $\ref{def3-1}$ with $\alpha \in (0,1)$ and
    \begin{align}\label{eq-u1u2}
    u_1(x) = \frac{1}{\alpha} \int_{0}^{x} v(t)dt,~u_2(x) = \frac{-1}{1-\alpha} \int_{-x}^{0} v(t)dt,~x\ge0.
    \end{align}
\end{enumerate}
\end{theorem}
\begin{proof}
\begin{enumerate}
\item[(i)]
From Proposition $\ref{pro-minimizer}$ and the convexity of $u_1,u_2$, $\rho_{\alpha}^{\mathcal{G}}(X)$ can be expressed as $$\essinf \left\{Z\in \mathcal{L} ^{\infty}(\mathcal{G}):\alpha \E[(u_1)_-^\prime((X-Z)^+)|\mathcal{G}]\le(1-\alpha)\E [(u_2)_+^\prime((X-Z)^-)|\mathcal{G}]\right\}.$$
As shown in Definition $\ref{def3-2}$ with $$v(x) = \begin{cases}
\alpha (u_1)_-^\prime(x), &  x > 0, \\
-(1-\alpha) (u_2)_+^\prime(-x), &  x \le  0,
\end{cases} $$
$$\rho_{v} ^\mathcal G(X) = \essinf\left\{Z\in \mathcal{L} ^{\infty}(\mathcal{G}):\E[v(X-Z)|\mathcal{G}] \le 0 \right\}, $$
which is actually the $\rho_{\alpha}^{\mathcal{G}}(X)$ here. $v$ is an increasing function because $u_1$ and $u_2$ are strictly increasing, convex functions on $\R_+$ with nonnegative increasing left derivatives and right derivatives.
\item[(ii)] Suppose \eqref{eq-u1u2} holds.
By the previous argument in (i) of the theorem, we have the following expression:
\begin{align*}
		\rho_{\alpha}^{\mathcal{G}}(X)&=\essinf \{Z\in \mathcal{L} ^{\infty}(\mathcal{G}): \E[\alpha(u_1)_-^\prime((X-Z)^+)|\mathcal{G}]\\ & \quad \quad \quad \quad -(1-\alpha)\E[(u_2)_+^\prime((X-Z)^-)|\mathcal{G}] \le 0 \}\\ &= \essinf \left\{Z\in \mathcal{L} ^{\infty}(\mathcal{G}): \E [v_0(X-Z)|\mathcal{G}]\le0\right\}\\
		&= \rho_{v_0}^{\mathcal{G}}(X)
		\end{align*}
with
$$
v_0(x) = \begin{cases}
\alpha (u_1)_-^\prime(x)=v(x-), & x>0,\\
-(1-\alpha) (u_2)_+^\prime(-x)=v(x-),& x\le 0.
\end{cases}$$
Since $v$ is increasing, we have $v(x)\ge v(x-)$ for all $x\in\R$.
It holds that
\begin{align*}
		\rho_{v_0}^{\mathcal{G}}(X) &= \essinf \left\{Z\in \mathcal{L} ^{\infty}(\mathcal{G}): \E v_0[(X-Z)|\mathcal{G}]\le 0\right\}\\&\le \essinf \left\{Z\in \mathcal{L} ^{\infty}(\mathcal{G}): \E v[(X-Z)|\mathcal{G}]\le 0\right\}\\&= \rho_{v}^{\mathcal{G}}(X).
	\end{align*}

Define $v_{\epsilon}(x)=v(x-\epsilon)$ with $\epsilon>0$, and we have $v_{\epsilon}(x)\le v(x-)\le v(x)$ for all $x\in\R$. It holds that $\rho_{v_{\epsilon}}^{\mathcal{G}}(X)\le \rho_{v_0}^{\mathcal{G}}(X)\le \rho_{v}^{\mathcal{G}}(X)$. On the other hand,
\begin{align*}
		\rho_{v_\epsilon}^{\mathcal{G}}(X) &= \essinf \left\{Z\in\mathcal L^{\infty}(\mathcal G): \E v_{\epsilon}(X-Z)\le 0\right\}\\
		&= \essinf \left\{Z-\varepsilon\in\mathcal L^{\infty}(\mathcal G): \E v(X-Z)\le 0\right\}\\
		&= \essinf \left\{Z\in\mathcal L^{\infty}(\mathcal G): \E v(X-Z)\le 0\right\}- \epsilon\\
		&= \rho_{v}^{\mathcal{G}}(X) - \epsilon.
\end{align*}
Hence, we have $\rho_{v_0}^{\mathcal{G}}(X)\ge \rho_{v}^{\mathcal{G}}(X) - \epsilon$ for any $\epsilon>0$. Furthermore, let $\epsilon\to 0$, and we obtain $\rho_{v_0}^{\mathcal{G}}(X)\ge \rho_{v}^{\mathcal{G}}(X)$. Therefore, we conclude that $\rho_\alpha^{\mathcal{G}}(X) = \rho_{v_0}^{\mathcal{G}}(X)= \rho_{v}^{\mathcal{G}}(X)$. This completes the proof.

%
\end{enumerate}
\end{proof}
Theorem \ref{th-eq} demonstrates the equivalence between conditional generalized quantiles and conditional shortfall risk measures, which provides us with a new perspective to study the properties of conditional generalized quantiles. In the following, we provide some examples of generalized conditional quantiles.
\begin{example} \begin{enumerate}[(i)]
		\item 	If $\mathcal{G} = \left\{\emptyset ,\Omega\right\}$, it means that there is no extra information provided, which reveals that there is no difference between generalized quantiles (\citealp{BB14}) and conditional generalized quantiles.

		\item If $\mathcal{G} = \mathcal{F}$, then the full information of  $X\in \mathcal L^\infty(\mathcal F)$ is available, which implies $\rho_{\alpha}^\mathcal{G}(X) = X$.

	\end{enumerate}
\end{example}

\begin{example}(Discrete case)
	Let $\Omega = \left\{\omega_1,\omega_2,\omega_3 \right\}$, $\p(\omega_i) = 1/3$, $X(\omega_i)=i$, $\alpha= 1/3$ and $u_1(x)= x^2, u_2(x)= e^x$. If $\mathcal{G}=  \left\{\emptyset,\Omega,A,A^c \right\}$ with $A = \left\{\omega_1, \omega_2 \right\}$, then we have
$$
\rho_{1/3}^\mathcal{G}(X)(\omega) = \left\{
	\begin{array}{rcl}
	1 ,& &{\omega \in A,}\\
	3, & &{\omega \in A^c,}
\end{array} \right.
$$
If $\mathcal{G} = \left\{\emptyset ,\Omega\right\}$, then
$\rho^{\mathcal G}_{1/3}(X)= a$ where $a$ is the solution of $e^{a-1} + 2a -5=0$, and one can calculate that $a\approx 1.594$.
	
\end{example}

\begin{example}(Conditional VaR and conditional expectile)
\begin{itemize}
\item[(i)] Let $u_1(x) = u_2(x) = x$. For a sub-$\sigma$-field $\mathcal G$, $\rho_{\alpha}^\mathcal{G}$ is a conditional quantile which can be represented as\footnote{
	Note that quantile as a risk measure is known as value-at-risk (VaR), and the conditional value-at-risk ($\mathrm{CVaR}$, also named $\mathrm{ES}$) is not the conditional quantile that we display here, although it has an ambiguous name.}
	\begin{align*}
	\rho_{\alpha}^\mathcal{G}(X)
	&= \essinf \left\{Z\in \mathcal{L} ^{\infty}(\mathcal{G}): \E [\alpha - \one_{\left\{X \le Z\right\}}|\mathcal{G}] \le 0\right\}\\
	& = \essinf \left\{Z\in \mathcal{L} ^{\infty}(\mathcal{G}):   \p(X \le Z |\mathcal{G})  \le \alpha\right\}.
	\end{align*}
\item[(ii)] Let $u_1(x) = u_2(x) = x^2$. We have
$$
    \rho_{\alpha}^\mathcal{G}(X) = \essinf \left\{ Z\in \mathcal{L} ^{\infty}(\mathcal{G}):\alpha \E [(X-Z)^+] +(1-\alpha)\E [(X-Z)^-]\le 0\right\},
$$
which is called conditional expectile in \cite{BB18}.
\end{itemize}
\end{example}

\begin{example}\label{entropy}
	(Conditional entropic risk measure) Let
$v_\gamma(x)= e^{\gamma x}-1,~\gamma >0$ and $v_\gamma(x)= 1-e^{\gamma x},~\gamma <0$.
	In this case, $\rho_{v}^{\mathcal{G}}$ is called the conditional entropic risk measure which is given as
	\begin{align*}
	\rho_{v_\gamma}^{\mathcal{G}}(X)
=\frac{1}{\gamma}\log \E [e^{\gamma X}|\mathcal{G}].
	\end{align*}
For the limiting cases of $\gamma=0$ and $\gamma=+\infty$, the conditional entropic risk measure can be represented as $\E[X|\mathcal G]$ and $\esssup [X|\mathcal G]$, respectively (\citealp{FS16}).
	One can verify that the conditional entropic risk measure with the parameter $\gamma\ge 0$ is a conditional convex shortfall risk measure (see Theorem \ref{th-cxcoherency}). Moreover, we will show that conditional entropic risk measures is the unique class of conditional generalized quantiles that satisfy time consistency (see Theorem \ref{th-timeconsistency}).
\end{example}

\section{Characterizations of conditional convexity and conditional coherency}\label{sec:characterization}
In this section, we aim to present the equivalent characterizations of the coherency and convexity of generalized conditional quantiles for the case in which the sub-$\sigma$-algebra $\mathcal G$ is generated by a finite partition of $\Omega$. Before showing the characterizations, we first explore some basic properties of the conditional generalized quantiles including monotonicity, conditional translation invariance, conditional convexity and conditional positive homogeneity.

\begin{proposition}\label{pro-basic}
Let
$\mathcal G \subseteq  \mathcal{F}$ be a $\sigma$-algebra, $u_1, u_2 \in \mathcal U_{icx}^{+}$ and $\alpha \in (0,1)$. The conditional generalized quantile $\rho_\alpha ^\mathcal G$ is defined in Definition $\ref{def3-1}$. The following statements hold.
\begin{enumerate}
\item[(i)] If $\alpha_{1}\le \alpha_{2}$, then for all $X\in\mathcal L^\infty(\mathcal F)$, $\rho_{\alpha_{1}}^{\mathcal{G}}(X) \le \rho_{\alpha_{2}}^{\mathcal{G}}(X)$.
\item[(ii)] Monotonicity: for all $X, Y \in \mathcal L^\infty(\mathcal{F})$, if $X \le Y$, then $\rho_{\alpha}^{\mathcal{G}}(X) \le \rho_{\alpha}^{\mathcal{G}}(Y)$.
\item [(iii)] Conditional translation invariance: for all $X\in\mathcal L^\infty(\mathcal F)$ and $H\in \mathcal L^\infty(\mathcal{G})$, $\rho_{\alpha}^{\mathcal{G}}(X+H) = \rho_{\alpha}^{\mathcal{G}}(X) + H$.

\item[(iv)] Normalization:
$\rho_{\alpha}^{\mathcal{G}}(0)=0$.

\item [(v)] Conditional convexity: if $u_1, u_2 \in \mathcal U_{icx}^{+}$ are twice differentiable, and $u_1^\prime$ is convex, $u_2^\prime$ is concave, and $\alpha u_1^{\prime\prime}(0)\ge (1-\alpha)u_2^{\prime\prime}(0)$,
then for all $X,Y \in \mathcal L^\infty(\mathcal{F})$ and $\Lambda\in \mathcal L^\infty(\mathcal{G})$ with $0 \le \Lambda \le1$, $\rho_{\alpha}^\mathcal{G}(\Lambda X+ (1-\Lambda)Y)\le \Lambda\rho_{\alpha}^\mathcal{G}(X)+ (1-\Lambda)\rho_{\alpha}^\mathcal{G}(Y)$.
\item[(vi)] Conditional positive homogeneity: if $u_1(x)= a_1x^\beta$, $u_2(x)= a_2x^\beta$ with $a_1, a_2 >0,\beta > 1$, then for all $X\in\mathcal L^\infty(\mathcal{F})$ and $\Lambda\in\mathcal L^\infty(\mathcal{G})$ with $\Lambda \ge 0$, $\rho_{\alpha}^\mathcal{G}(\Lambda X)= \Lambda\rho_{\alpha}^\mathcal{G}(X)$.
\end{enumerate}
\end{proposition}
\begin{proof}
\begin{enumerate}
\item[(i)] From Theorem $\ref{th-eq}$, we have

$$\rho_\alpha^{\mathcal{G}}(X) =\rho_v^{\mathcal{G}}(X) = \essinf \left\{Z\in \mathcal{L} ^{\infty}(\mathcal{G}): \E [v(X-Z)|\mathcal{G}]\le 0\right\}$$ with an increasing function
$$v(x) = \begin{cases}
\alpha (u_1)_-^\prime(x), & x>0,\\
-(1-\alpha) (u_2)_+^\prime(-x) ,& x\le 0.
\end{cases}$$
Since $\alpha_1\le\alpha_2$, and $(u_1)_-^\prime,~(u_2)_+^\prime$ are nonnegative, we have $\alpha_1 (u_1)_-(x) \le \alpha_1 (u_1)_-(x)$ for all $x>0$ and $-(1-\alpha_1)(u_2)_+^\prime(-x)\le
-(1-\alpha_2)(u_2)_+^\prime(-x)$ for all $x\le0$. Hence,
we have $\rho_{\alpha_{1}}^{\mathcal{G}}(X) \le \rho_{\alpha_{2}}^{\mathcal{G}}(X)$.
\item[(ii)] Since $X \le Y$ implies $v(X-Z)\le v(Y-Z)$, we have
\begin{align*}
		 	\rho_\alpha^{\mathcal{G}}(X) =\rho_v^{\mathcal{G}}(X) &= \essinf \left\{Z\in \mathcal{L} ^{\infty}(\mathcal{G}): \E [v(X-Z)|\mathcal{G}]\le 0\right\}\\
		 	&\le \essinf \left\{Z\in \mathcal{L} ^{\infty}(\mathcal{G}): \E [v(Y-Z)|\mathcal{G}]\le 0\right\}\\
		 	&= \rho_v^{\mathcal{G}}(Y)=\rho_\alpha^{\mathcal{G}}(Y).
		 \end{align*}
\item[(iii)]From Theorem $\ref{th-eq}$, for any $H\in \mathcal L^\infty(\mathcal{G})$,
\begin{align*}
	      	\rho_\alpha^{\mathcal{G}}(X+H) =\rho_v^{\mathcal{G}}(X+H) &= \essinf \left\{Z\in \mathcal{L} ^{\infty}(\mathcal{G}): \E [v(X-(Z-H))|\mathcal{G}]\le 0\right\}\\ &= \essinf \left\{Z+H\in \mathcal{L} ^{\infty}(\mathcal{G}): \E [v(X-Z)|\mathcal{G}]\le 0\right\}\\ &= \essinf \left\{Z\in \mathcal{L} ^{\infty}(\mathcal{G}): \E [v(X-Z)|\mathcal{G}]\le 0\right\} + H\\ &= \rho_\alpha^{\mathcal{G}}(X) +H.
	      \end{align*}

\item[(iv)] From Definition \ref{def3-1}, we have
\begin{align*}
	\rho_\alpha ^\mathcal G(0) &= \essinf\argmin \limits_{Z\in \mathcal{L} ^{\infty}(\mathcal{G})} \pi _\alpha (0,Z)\\
&=\essinf\argmin \limits_{Z\in \mathcal{L} ^{\infty}(\mathcal{G})} \alpha \E[ u_1((-Z)^+)] +(1-\alpha)\E [u_2((-Z)^-)].
 \end{align*}
Note that $u_1$ and $u_2$ are nonnegative with $u_1(0)=u_2(0)=0$. Hence, one can verify that the minimizer of above equation is attained at $Z=0$. Hence, we have $\rho_\alpha ^\mathcal G$ is normalized.
\item[(v)]
We first verify that the convexity of the function $v$ defined by \eqref{eq-v} implies the conditional convexity of $\rho_\alpha^{\mathcal G}$.
From the convexity of $v$, it holds that  for $X,Y \in \mathcal L^\infty(\mathcal{F})$ and $\Lambda\in \mathcal L^\infty(\mathcal{G})$ with $0 \le \Lambda \le 1$, $$\E [v(\Lambda X+ (1-\Lambda)Y-Z)|\mathcal{G}]\le \Lambda \E [v(X-Z)|\mathcal{G}] +(1-\Lambda)\E [v(Y-Z)|\mathcal{G}],$$ and thus, we have
\begin{align*}
	      &\left\{Z\in \mathcal{L} ^{\infty}(\mathcal{G}): \E [v(\Lambda X+ (1-\Lambda)Y-Z)|\mathcal{G}]\le 0\right\}\supseteq\\&  \left\{Z\in \mathcal{L} ^{\infty}(\mathcal{G}): \Lambda \E [v(X-Z)|\mathcal{G}] +(1-\Lambda)\E [v(Y-Z)|\mathcal{G}]\le 0\right\},
	      \end{align*}
which follows that
\begin{align*}
	      \rho_{v}^\mathcal{G}(\Lambda X+ (1-\Lambda)Y)& = \essinf \left\{Z\in \mathcal{L} ^{\infty}(\mathcal{G}): \E [v(\Lambda X+ (1-\Lambda)Y-Z)|\mathcal{G}]\le 0\right\}\\
	      &\le \essinf\{\Lambda Z_1+(1-\Lambda)Z_2 \in \mathcal{L} ^{\infty}(\mathcal{G}): \Lambda \E [v(X-Z_1)|\mathcal{G}]+\\
	      	      & \quad \quad \quad \quad \quad(1-\Lambda)\E [v(Y-Z_2)|\mathcal{G}]\le 0 \} \\
	      &\le \Lambda \essinf \{Z_1\in \mathcal{L} ^{\infty}(\mathcal{G}): \E [v(X-Z_1)|\mathcal{G}]\le 0\}+\\
	      & \quad(1-\Lambda) \essinf \left\{Z_2\in \mathcal{L} ^{\infty}(\mathcal{G}): \E [v(Y-Z_2)|\mathcal{G}]\le 0\right\}\\&=  \Lambda \rho_{v}^\mathcal{G}(X)+ (1-\Lambda)\rho_{v}^\mathcal{G}(Y).
	      \end{align*}
Hence, we conclude that the convexity of $v$ implies the conditional convexity of $\rho_\alpha^{\mathcal G}$. Finally, note that the conditions of $u_1,u_2$ in (v) are sufficient for the convexity of $v$. This completes the proof of (v).

\item[(vi)]
From Theorem $\ref{th-eq}$, we obtain
$\rho_\alpha^{\mathcal{G}}(X) = \essinf \left\{Z\in \mathcal{L} ^{\infty}(\mathcal{G}): \E [v(X-Z)|\mathcal{G}]\le 0\right\}$ with 
$$v(x) = \begin{cases}
\alpha \beta a_1 x^{\beta-1}, & x>0,\\
(1-\alpha) \beta a_2(-x)^{\beta-1} ,& x\le 0.
\end{cases}$$
Thus for all $\Lambda\in\mathcal L^\infty(\mathcal{G})$ with $\Lambda \ge 0$, we have
\begin{align*}
	       	\rho_{\alpha}^\mathcal{G}(\Lambda X)&= \essinf \left\{Z\in \mathcal{L} ^{\infty}(\mathcal{G}): \E [v(\Lambda X-Z)|\mathcal{G}]\le 0\right\}\\
	       	& = \essinf \{Z\in \mathcal{L} ^{\infty}(\mathcal{G}): \E [ \alpha \beta a_1 ((\Lambda X-Z)^+)^{\beta-1}+\\
	       	&\quad \quad \quad \quad \quad(1-\alpha) \beta a_2((\Lambda X-Z)^-)^{\beta-1}|\mathcal{G}]\le 0\}\\
	       	&= \essinf \{\Lambda Z \in \mathcal{L} ^{\infty}(\mathcal{G}): \Lambda^{\beta-1}\E [ \alpha \beta a_1 (( X-Z)^+)^{\beta-1}+\\
	       	&\quad \quad \quad \quad \quad(1-\alpha) \beta a_2(( X-Z)^-)^{\beta-1}|\mathcal{G}]\le 0\}\\
	       	&= \Lambda \essinf \{Z\in \mathcal{L} ^{\infty}(\mathcal{G}): \E [ \alpha \beta a_1 (( X-Z)^+)^{\beta-1}+\\
	       	&\quad \quad \quad \quad \quad(1-\alpha) \beta a_2(( X-Z)^-)^{\beta-1}|\mathcal{G}]\le 0\}\\
	       	&= \Lambda\rho_{\alpha}^\mathcal{G}(X).
	       \end{align*}
\end{enumerate}
This completes the proof.
\end{proof}

\begin{theorem}\label{th-cxcoherency}
Let $u_1, u_2 \in \mathcal U_{icx}^{+}$ be twice differentiable and strictly convex functions with $u_1^\prime(0)= u_2^\prime(0)= 0$ and $\alpha \in (0,1)$. The conditional generalized quantile $\rho_\alpha ^\mathcal G$ is defined as in Definition $\ref{def3-1}$. We have
\begin{itemize}
\item[(i)]
If $u_1^\prime$ is convex, $u_2^\prime$ is concave, and $\alpha u_1^{\prime\prime}(0)\ge (1-\alpha)u_2^{\prime\prime}(0)$, then $\rho_{\alpha}^\mathcal{G}$ is a conditional convex risk measure. Moreover, if $\mathcal G \subseteq  \mathcal{F}$ is generated by a finite partition of $\Omega$, then the converse direction holds.

\item[(ii)]

If $u_1(x)= x^2$, $u_2(x)= x^2$ and $\alpha \ge1/2$, then $\rho_{\alpha}^\mathcal{G}$ is a conditional coherent risk measure. Moreover, if $\mathcal G \subseteq  \mathcal{F}$ is generated by a finite partition of $\Omega$, then the converse direction holds.

\end{itemize}
\end{theorem}

\begin{proof}
(i)
By Proposition \ref{pro-basic} (ii), (iii), (iv) and (v), we have $\rho_{\alpha}^\mathcal{G}$ is monotonic, conditional translation invariant, normalized and conditional convex which means it is a conditional convex risk measure.
To see the ``Moreover" part, suppose that $\mathcal G=\sigma(A_1,A_2,\cdots,A_n)$, where $\{A_1,A_2,\cdots,A_n\}$ is the partition of $\Omega$ with $\p(A_i)>0$ for $i=1,2,\cdots,n$.
Denote by $F_{X|\mathcal{G}}(\cdot, \omega)$ the regular conditional distribution of $X$ on $\mathcal{G}$. It is clear that for $\omega\in A_i$,
$$
F_{X|\mathcal{G}}(x, \omega)=\p(X\le x|A_i).
$$
Then, the conditional risk measure $\rho_\alpha^{\mathcal G}(X)$ can be obtained in the following way:
$$
\rho_\alpha^{\mathcal G}(X)(\omega)=\inf\left\{z: \int_{\R}v(x-z){\d} F_{X|\mathcal{G}}(x, \omega)\right\}.
$$
Define a convex set of random variables:
$$
\mathcal X=\{X\in\mathcal L^{\infty}(\mathcal F): X(\omega)=0~{\rm for~all}~\omega\in\cup_{i=2}^n A_i\}.
$$
For any $X\in\mathcal X$, one can obtain
$$
\rho_\alpha^{\mathcal G}(X)(\omega)=0,~{\rm for}~\omega\in\cup_{i=2}^n A_i,
$$
and
$$
\rho_\alpha^{\mathcal G}(X)(\omega)=\inf\left\{z: \int_{\R}v(x-z){\d} \p(X\le x|A_1)\right\},~{\rm for}~\omega\in A_1.
$$
The convexity of $\rho_\alpha^{\mathcal G}$ implies that for $X_1,X_2\in\mathcal X$ and $\lambda\in(0,1)$,
$$
\rho_\alpha^{\mathcal G}(\lambda X_1+(1-\lambda) X_2)(\omega)\le \lambda\rho_\alpha^{\mathcal G}( X_1)(\omega)+(1-\lambda)\rho_\alpha^{\mathcal G}( X_2)(\omega),~{\rm for}~\omega \in A_1,
$$
which is equivalent to
$$
\rho_\alpha(F_{\lambda X_1+(1-\lambda) X_2}^{A_1})\le \lambda \rho_\alpha(F_{X_1}^{A_1})+(1-\lambda) \rho_\alpha(F_{X_2}^{A_1}),
$$
where $\rho_\alpha(F):=\inf\left\{z:\int_{\R}v(x-z){\d}F(x)\right\}$ is the classic shortfall risk measure and $F_X^{A_1}(x):=\p(X\le x|A_1)$ is the conditional distribution function. Applying the proof of Corollary 3.1 in \cite{Weber06}, we finally obtain the desired result.

(ii) Suppose $u_1(x)= x^2$, $u_2(x)= x^2$ and $\alpha \ge1/2$. It follows from Proposition \ref{pro-basic} (i), (ii), (iii) and (iv) that $\rho_{\alpha}^\mathcal{G}$ is a conditional coherent risk measure.
By the previous argument of (i), the ``Moreover" part follows immediately from the proof of Corollary 3.2 in \cite{Weber06}.

\end{proof}

\begin{remark} It is worth noting that
the assumption that $\mathcal G$ is generated by a finite partition of
$\Omega$ is necessary in the converse direction.
We illustrate this point by giving a counter-example here. Suppose $\mathcal G=\mathcal F$, and we have $\rho_{\alpha}^{\mathcal G}(X)=X$ for all $X\in\mathcal L^{\infty}(\mathcal F)$. This means that for arbitrary $u_1,u_2\in\mathcal U_{icx}^+$, $\rho_{\alpha}^{\mathcal G}$ is always a conditional coherent risk measure.

The conditional risk measure in Theorem \ref{th-cxcoherency} (ii)
is called the conditional expectile introduced by \cite{BB18},
and has been discussed in detail. By Theorem \ref{th-cxcoherency}, we conclude that conditional expectile is the only class of conditional coherent risk measures in the class of conditional generalized quantiles if $\mathcal G$ is generated by a finite partition of $\Omega$.
\end{remark}

\section{Further properties}\label{Properties}
In this section, we first explore some futher properties of the conditional generalized quantiles, and then investigate the risk measure under the dynamic framework.

\begin{theorem}\label{pro-further}
Let $X \in\mathcal{L}^\infty(\mathcal{F})$, $\mathcal G \subseteq  \mathcal{F}$ be a $\sigma$-algebra, $u_1, u_2 \in \mathcal U_{icx}^{+}$ and $\alpha \in (0,1)$. $\rho_\alpha ^\mathcal G(X)$ is defined as in Definition $\ref{def3-1}$. The following statements hold.
\begin{enumerate}
\item[(i)] $\rho_\alpha ^\mathcal G(X)$ satisfies $\rho_\alpha ^\mathcal G(X)(\omega)= \rho_\alpha(F_\mathcal{G}(\cdot, \omega))$, where $F_\mathcal{G}(\cdot, \omega)$ is a regular conditional distribution of X on $\mathcal{G}$.
\item[(ii)] If $X_n\uparrow X$ as $n\to\infty$, then $ \rho_\alpha^\mathcal{G}(X_n)\uparrow \rho_\alpha^\mathcal{G}(X)$ as $n\to\infty$.

\end{enumerate}
\end{theorem}

\begin{proof}
(i) From Theorem $\ref{th-eq}$, we obtain the following equation:
\begin{align*}
	\rho_\alpha^{\mathcal{G}}(X) &= \essinf \left\{Z\in \mathcal{L} ^{\infty}(\mathcal{G}): \E[ v(X-Z)|\mathcal{G}]\le 0\right\}.
\end{align*}
Then, for almost every $\omega \in \Omega$,
$$\rho_\alpha^{\mathcal{G}}(X)(\omega) = \inf \left\{z: \E[ v(Y-z)]\le 0\right\},$$
where $Y$ is a random variable with distribution $F_\mathcal{G}(\cdot, \omega)$.
Thus, $\rho_\alpha ^\mathcal G(X)(\omega)= \rho_\alpha(F_\mathcal{G}(\cdot, \omega))$, where $F_\mathcal{G}(\cdot, \omega)$ is a regular conditional distribution of X on $\mathcal{G}$.

(ii) Since $X_n\uparrow X$, we have that $\rho_\alpha^\mathcal{G}(X_n) \leq \rho_\alpha^\mathcal{G}(X_{n+1})$. Let $\rho_\alpha^\mathcal{G}(X_n)\uparrow Z$ with $Z\leq \rho_\alpha^\mathcal{G}(X)$.
From the monotonicity of $v$, it follows that
$$v(X_n-\rho_\alpha^\mathcal{G}(X_n))\leq v(X-\rho_\alpha^\mathcal{G}(X_1)).$$
By the dominated convergence theorem, we have
$$0 = \E [v(X_n-\rho_\alpha^\mathcal{G}(X_n))|\mathcal{G}] \rightarrow
\E [v(X-Z)|\mathcal{G}].$$
It follows that $\rho_\alpha^\mathcal{G}(X)= Z$, and thus, $\rho_\alpha^\mathcal{G}(X_n)\uparrow \rho_\alpha^\mathcal{G}(X)$ as $n\to\infty$.
{}
\end{proof}

Theorem \ref{pro-further} (i)
states that the conditional generalized quantile is the static generalized quantile given the information in $\mathcal{G}$, which reveals the conditional distribution invariance because $\rho_\alpha^{\mathcal{G}}(X)$ only depends on the conditional distribution of $X$, and (ii) illustrates that the risk measure is continuous for monotonically convergent sequence.

Next, we investigate conditional generalized quantiles in a dynamic framework. Consider a set of dates $\mathcal T=[0,T]$ when the riskiness of a final payoff at time $T$ is assessed. We introduce a filtration $\{\mathcal F_t\}_{t\in\mathcal T}$ where $\mathcal F_t$ represents the information available at time $t$. Moreover, suppose that $\mathcal F_0=\{\emptyset,\Omega\}$ and $\mathcal F_T=\mathcal F$. We recall the definition of \emph{dynamic risk measure}. For the detailed introduction of dynamic risk measure, we refer to \cite{A07}, \cite{CDK06}, \cite{FP06} and \cite{KS09} and the references therein.

\begin{definition}[Definition 8 of \citealp{DS05}]
A dynamic risk measure is a family $\{\rho^{\mathcal F_n}\}_{n=0}^N$ such that $\rho^{\mathcal F_n}: \mathcal L^\infty(\mathcal F)\to \mathcal L^\infty({\mathcal F_n})$ is a conditional risk measure for all $0\le n\le N$. It is a dynamic convex (coherent) risk measure if all components $\rho^{\mathcal F_n}$ are conditional convex (coherent) risk measures.
\end{definition}

The following proposition is a direct result from Theorem \ref{th-cxcoherency}.

\begin{proposition}
Let $u_1, u_2 \in \mathcal U_{icx}^{+}$ be twice differentiable and strictly convex functions with $u_1^\prime(0)= u_2^\prime(0)= 0$ and $\alpha \in (0,1)$. The conditional generalized quantile $\rho_\alpha ^\mathcal G$ is defined as in Definition $\ref{def3-1}$. We have
\begin{itemize}
\item[(i)]
If $u_1^\prime$ is convex, $u_2^\prime$ is concave, and $\alpha u_1^{\prime\prime}(0)\ge (1-\alpha)u_2^{\prime\prime}(0)$, then $\{\rho_{\alpha}^{\mathcal F_n}\}_{n=0}^{N}$ is a dynamic convex risk measure.

\item[(ii)]
If $u_1(x)= x^2$, $u_2(x)= x^2$ and $\alpha \ge 1/2$, then $\{\rho_{\alpha}^{\mathcal{F}_n}\}_{n=0}^{N}$ is a dynamic coherent risk measure.
\end{itemize}
\end{proposition}

The property of time consistency, usually referring to sequential consistency, supermartingale property and dynamic consistency, plays a crucial role in dynamic decision-making problem. We collect the definition of the property below (see e.g., \citealp{AP11}; \citealp{BB18}).
\begin{definition}
	A dynamic risk measure $\{\rho^{\mathcal F_n}\}_{n=1}^N$
is said
\begin{enumerate}[(i)]
\item \emph{sequentially consistent}, if for all $X\in\mathcal{L}^\infty(\mathcal{F})$ and
    $t\in\mathcal T$,
  $$
  \rho^{\mathcal{F}_0} (X) \le 0 \Rightarrow   \rho^{\mathcal{F}_{t}}(X)\le 0 \text{ and } \rho^{\mathcal{F}_0}(X)\ge 0 \Rightarrow  \rho^{\mathcal{F}_{t}}(X)\ge 0.
  $$

\item \emph{dynamic consistent},
if for all $X, Y\in\mathcal{L}^\infty(\mathcal{F})$ and $t\in\mathcal T$,
    $$
    \rho^{\mathcal F_{t}}(X)=\rho^{\mathcal F_{t}}(Y)\Rightarrow \rho^{\mathcal F_{0}}(X)=\rho^{\mathcal F_{0}}(Y).
    $$

\item  to have the \emph{supermartingale property}, if for all $X\in\mathcal{L}^\infty(\mathcal{F})$ and $t\in\mathcal T$, $$
    \rho^{\mathcal F_0}(X)\ge\E[\rho^{\mathcal{\mathcal F}_t}(X)].
    $$
 \end{enumerate}
\end{definition}
Sequential consistency describes  that the conclusions  made with more information should still hold with less information.  We will show later that the dynamic generalized quantiles is sequentially consistent. Dynamic consistency is stronger than sequential consistency, and it is equivalent to \emph{tower property} (\citealp{DS05}), that is, $\rho^{\mathcal F_0}(\rho^{\mathcal{F}_t}(X))=\rho^{\mathcal F_0}(X)$, for monetary and law invariant dynamic risk measures\footnote{We say $\{\rho^{\mathcal F_t}\}_{t\in\mathcal T}$ is monetary if $\rho^{\mathcal F_t}$ is monotonic and translation invariant for all $t\in\mathcal T$.
The law invariant property means that
$\rho^{\mathcal F_0}(X)=\rho^{\mathcal F_0}(Y)$ whenever $X,Y$ has the same distribution function.}. \cite{KS09} proved that the dynamic entropic risk measure (see Example \ref{entropy}) is the unique class of  monetary and law invariant dynamic risk measures that satisfy dynamic consistency. Similarly, we will show that supermartingale property characterizes dynamic entropic risk measure among dynamic generalized quantiles.


\begin{theorem}\label{th-timeconsistency}
\begin{enumerate}[(i)]
\item   
A dynamic generalized quantile $\{\rho_\alpha^{\mathcal{F}_n}\}_{t\in\mathcal T}$ defined in Definition $\ref{def3-1}$ is sequentially consistent.
 \item 
A dynamic generalized quantile $\{\rho_\alpha^{\mathcal{F}_n}\}_{t\in\mathcal T}$ is dynamic consistent if and only if it is
a dynamic entropic risk measure defined in Example \ref{entropy}.
\item 
A dynamic generalized quantile $\{\rho_\alpha^{\mathcal{F}_n}\}_{t\in\mathcal T}$ has the supermartingale property if and only if it is a
dynamic entropic risk measure defined in Example \ref{entropy} with the parameter $\gamma\ge 0$.

\end{enumerate}
\end{theorem}
\begin{proof}
(i) 
By Theorem \ref{th-eq},  $\rho_\alpha^{\mathcal{F}_{t}}(X) \le 0 $ implies $\E[ v(X)|\mathcal{F}_{t}]\le 0$ where $v$ is defined by \eqref{eq-v}.
It holds that $$ \E[ v(X)|\mathcal{F}_0] =\E[\E[ v(X)|\mathcal{F}_{t}]|\mathcal F_0]\le 0,$$
and thus, $\rho_\alpha^{\mathcal{F}_0}(X) \le 0$.
The other inequality follows similarly.

(ii)
%
%
By Proposition \ref{pro-basic} (ii) and (iii),
the dynamic generalized quantiles are all monetary.
The law invariance of $\rho_\alpha^{\mathcal F_0}$ is trivial.
Hence, the result follows immediately from Theorem 1.10 of \cite{KS09} which states that the dynamic entropic risk measure is the unique class of
monetary and law invariant dynamic risk measures that satisfy dynamic consistency.

(iii) One can verify that conditional entropic risk measure has the supermartingale property (see e.g., Proposition 6 of \citealp{DS05}).

To see necessity,
we first recall the definition of the mixture concavity and convex level set (CxLS). Noting that $\rho_{\alpha}^{\mathcal F_0}$ is law invariant, we denote by $\rho_\alpha^{\mathcal F_0}(F)=\rho_\alpha^{\mathcal F_0}(X)$ with $X\sim F$.
For two distribution functions $F,G$,
the \emph{mixture concavity} of unconditional generalized quantile $\rho_\alpha^{\mathcal F_0}$ is shown as $$\rho_\alpha^{\mathcal F_0}(\lambda F+(1-\lambda)G))\ge\lambda\rho_\alpha^{\mathcal F_0}(F)+(1-\lambda)\rho_\alpha^{\mathcal F_0}(G),~~\lambda\in(0,1).$$
We say $\rho_\alpha^{\mathcal F_0}$ has CxLS if and only if $\rho_\alpha^{\mathcal F_0}(F)=\rho_\alpha^{\mathcal F_0}(G)$ implies $\rho_\alpha^{\mathcal F_0}(\lambda F+(1-\lambda G))=\rho_\alpha^{\mathcal F_0}(F)$ for all $\lambda\in(0,1)$.
Note that the supermartingale property implies the mixture concavity of $\rho_\alpha^{\mathcal F_0}$ (see e.g.,
\citealp{PR07}).
Further, by Theorem 4 of \cite{EM21}, the entropic risk measure  is the unique class of law-invariant,  monetary and normalized risk measures that satisfy mixture concavity and have CxLS.
Since all generalized quantiles are law-invariant, monetary and normalized risk measures and have CxLS (see Proposition \ref{pro-basic} (ii), (iii) and (iv)),
 we have that the desired result holds.
\end{proof}

\end{document}